\documentclass[12pt]{article}
\usepackage{amsmath}

\usepackage{enumerate}
\usepackage{url} 
\usepackage{hyperref}
\usepackage{amsfonts}
\usepackage{listings}
\usepackage{pagecolor}
\usepackage{xcolor}
\usepackage{lipsum}
\usepackage{amsthm}
\usepackage{graphicx,psfrag,epsf}

\newtheorem{theorem}{Theorem}

\newtheorem{proposition}[theorem]{Proposition}

\pagecolor{black}
\color{white}

\definecolor{codegreen}{rgb}{0,0.6,0}
\definecolor{codegray}{rgb}{0.5,0.5,0.5}
\definecolor{codepurple}{rgb}{0.58,0,0.82}
\definecolor{backcolour}{rgb}{0.95,0.95,0.92}

\lstdefinestyle{mystyle}{
    backgroundcolor=\color{backcolour},
    commentstyle=\color{codegreen},
    keywordstyle=\color{magenta},
    numberstyle=\tiny\color{codegray},
    stringstyle=\color{codepurple},
    basicstyle=\ttfamily\footnotesize,
    breakatwhitespace=false,         
    breaklines=true,
    captionpos=b,
    keepspaces=true,
    numbers=left,
    numbersep=5pt,
    showspaces=false,
    showstringspaces=false,
    showtabs=false,
    tabsize=2
}

\newcommand{\blind}{1}

\pagecolor{white}
\color{black}

\begin{document}

\def\spacingset#1{\renewcommand{\baselinestretch}%
{#1}\small\normalsize} \spacingset{1}


\if1\blind
{
  \title{\bf Reducing bias and variance in quantile estimates with an exponential model}
  \author{Rohit Pandey\hspace{.2cm}\\
    Principal Data Scientist, Microsoft Azure\\
    }
  \maketitle
} \fi

\if0\blind
{
  \bigskip
  \bigskip
  \bigskip
  \begin{center}
    {\LARGE\bf Reducing bias and variance in quantile estimates with an exponential model}
\end{center}
  \medskip
} \fi

\bigskip
\begin{abstract}
Percentiles and more generally, quantiles are commonly used in various contexts to summarize data. For most distributions, there is exactly one quantile that is unbiased. For distributions like the Gaussian that have the same mean and median, that becomes the medians. There are different ways to estimate quantiles from finite samples described in the literature and implemented in statistics packages. It is possible to leverage the memory-less property of the exponential distribution and design high quality estimators that are unbiased and have low variance and mean squared errors. Naturally, these estimators out-perform the ones in statistical packages when the underlying distribution is exponential. But, they also happen to generalize well when that assumption is violated.
\end{abstract}

\noindent%
{\it Keywords:}  percentiles, quartiles, bias-removal, median, memory-less
\vfill

\newpage
\spacingset{1.45} 

\section{Introduction}
\label{sec:intro}

Let's establish some definitions and notation. In doing so, I'll lean heavily on \cite{stat_pckgs}. The definition of the quantile of a distribution is as follows:

$$Q(p) = F^{-1}(p) = \inf\{x: F(x) \geq p\}, \;\;\;\;0<p<1$$

with $F(x)$ being the cumulative distribution function (CDF) of the distribution. The special case, $Q(0.5)$ is the median of the distribution.

In practical applications, we observe a finite number of (say $n$) samples from the underlying distribution, $\{X_1, X_2, \dots, X_n\}$ and set out to estimate $Q(p)$ from them. There are many ways of doing this, leading to different estimators for the same underlying property. Let's call the result of the $i$-th method for doing the estimator $\hat{Q}_i(p)$. 

The sample quantiles in most statistical packages are based on one or two order statistics and can be expressed as:

$$\hat{Q}_i(p) = (1-f)X_{(j)}+f X_{(j+1)}$$

For some $1\leq j \leq n$ and $0 \leq f \leq 1$. In \cite{harrel} and \cite{sheather}, estimators that don't conform to the format above are discussed. However, these are not implemented in popular statistical packages. In some of the new estimators I'll introduce, the condition on $f$ will be relaxed to $f \in \Bbb{R}$.

In a lot of instances, quantiles are used on distributions supported on $[0, \infty)$. For instance, household incomes, heights of people, etc. In the context of my work in cloud computing, latencies for operations like starting machines, blackout durations when performing certain kinds of servicing operations are other examples. In a lot of instances, such distributions tend to be heavy tailed with decreasing hazard rates. In this paper, we'll primarily be concerned with how the various estimators do on such distributions.

To measure the performance of estimators, two important criterion are the bias and variance. These can be combined into a single measure, the mean squared error (MSE). No matter what methodology we use, it is near impossible to eliminate the bias for quantiles across all underlying distributions and all quantiles (like it is when using the sample mean). 

In this paper, we'll first measure the performance of the nine strategies described in \cite{stat_pckgs} on these criterion for various underlying distributions. These nine strategies are implemented in the same order as the paper in the statistical programming language, R (see \href{https://www.rdocumentation.org/packages/stats/versions/3.6.2/topics/quantile}{the documentation}). Then, we'll introduce a new family of estimators that are based on the exponential distribution and pit them against the winning estimators and find their performance to be favorable on values of $p<0.6$. For the purpose of this study, I've created a new Python package (\href{https://github.com/ryu577/statest}{statest}) which has routines implementing all $9$ estimators from R, the new strategies based on the exponential distribution as well as simulators that measure the performance (bias, variance, MSE) of these strategies on various underlying distributions via Monte Carlo simulation.

\section{Existing quantile estimators in statistical packages}
For measuring the performance of the existing nine implementations, I tried to pick a range of distributions with a wide range of tail behaviors. The following is a list as of this writing: Normal, LogNormal (heavy tailed), Lomax (so heavy tailed the mean doesn't exist), LogLogistic (again, mean doesn't exist), Exponential (memory-less), Weibull (lighter tailed than exponential). In general, the conclusions for any distribution were largely invariant of the parameters of the distribution. So, I won't be specifying what parameters were used for each of them. Anyone that has access to a basic Python environment with the libraries `numpy' and `scipy' can re-create all the simulations and raw data (see readme file of the \href{https://github.com/ryu577/statest}{statest} library).

Figure \ref{fig:r_weibull_biases} shows the biases of the nine estimators available in R (and now Python thanks to the \href{https://github.com/ryu577/statest}{statest library}) on the Weibull distribution. The conclusion is similar for all other distributions. We can rule out estimators other than 1, 2, 3, 4 and 7 due to their very poor bias. And since 1, 2 and 3 have a discountinuous behavior with the quantile ($q$) being estimated (without doing substantially better on the bias criterion for any underlying distribution tested), we won't consider them as well.

\begin{figure}
\begin{center}
\includegraphics[height=5in, width=5in]{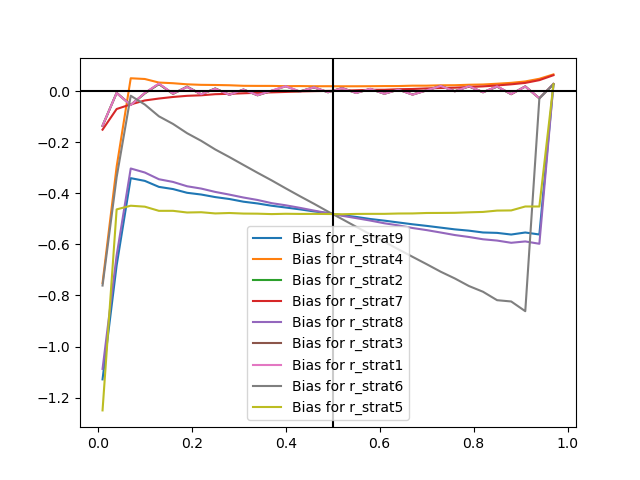}
\end{center}
\caption{The bias of the various estimators available in R on a sample size of 15 for the Weibull distribution. \label{fig:r_weibull_biases}}
\end{figure}

This leaves us with estimators 4 and 7. The definition of estimator 4 can be found in \cite{stat_pckgs}. Since we're going to motivate a new estimator based on estimator 7 and its also the default estimator in all statistical packages, I will briefly describe it.

\subsection{Estimator 7: default quantile estimator in statistical packages}

The most familiar estimator, which is the default in most libraries (the default behavior of R's quantile function and Python numpy's percentile function for example) can be motivated by the median. If our sample has $n=2l+1$ data points, the median should just be the middle order statistic, $X_{(l+1)}$. If if has $2l$ data points, there are two "middle order statistics", $X_{(l)}$ and $X_{(l+1)}$. So, we take the mid-point or average of them. In other words, we linearly interpolate between them. This is a logical strategy that will appeal to most people's intuition. Extending it to the $q$-percentile, it's not too hard to see that the order statistics, $X_{(i)}$ and $X_{(i+1)}$ we'll be interpolating between satisfy:

\begin{align}
\lfloor i-1 \rfloor &= \lfloor q(n-1) \rfloor \notag\\
=> i &= \lfloor q(n-1) \rfloor + 1 \label{eqn:ord_i}
\end{align}

Since this is the seventh strategy in \cite{stat_pckgs}, we'll call it $\hat{Q}_7(q)$. The linear interpolation gives us - 

$$\frac{X_{(i+1)}-X_{(i)}}{\left( \frac{1}{n-1} \right)} = \frac{\hat{Q}_7(q)-X_{i}}{q-\frac{\lfloor q(n-1)\rfloor}{n-1}}$$

$$=>\hat{Q}_7(q)=X_{(i+1)}(q(n-1)-\lfloor q(n-1)\rfloor)-X_{(i)}(q(n-1)-\lfloor q(n-1)\rfloor-1)$$

\begin{equation}=>\hat{Q}_7(q)=X_{(i+1)} \{q(n-1)\} + X_{(i)}(1-\{q(n-1)\}) \label{eqn:u}\end{equation}

Here, $x-\lfloor x \rfloor = \{x\}$ is the fractional part of $x$. Among the nine estimators described in \cite{stat_pckgs} (and implemented in the same order in R), this one appears as estimator 7. We'll refer to it simply as the `seventh estimator' going forward. 

\section{Measuring bias for the exponential distribution}

Since we're concerned with estimating quantiles for distributions that have support on $[0,\infty)$, it makes sense to start with the simplest such distribution, the Exponential distribution whose memory-less property gives us the ability to answer many questions about it analytically. 

\begin{proposition}
The order statistic of the exponential distribution (courtesy Renyi, see \cite{renyi}) is given by:

\begin{equation}X_{(i)} = \sum\limits_{j=1}^i \frac{Z_j}{n-j+1}\label{exp_ord}\end{equation}

Here, the $Z_j$ are i.i.d exponential themselves with rate $1$. 
\end{proposition}
\begin{proof}
To see this, consider the minimum, $X_{(1)}$. It's the minimum of $n$ exponentials. And this is just another exponential with rate $n \lambda$. 

$$X_{(1)} \sim \frac{Z_1}{n \lambda}$$

Now, we're at time $X_{(1)}$, and there are $n-1$ exponentials to go. By the memory-less property, the additional time it'll take for the next exponential to arrive from here is simply the minimum of $n-1$ exponentials. So,

$$X_{(2)}\sim \frac{Z_1}{n \lambda} + \frac{Z_2}{(n-1) \lambda}.$$

And we can take this all the way to $i$.
\end{proof}

We can use equations ~\ref{eqn:u} and ~\ref{exp_ord} to get an expression for $\hat{Q}_7(q)$ for the case of the exponential distribution:

\begin{equation}
\hat{Q}_7(q) =  \sum\limits_{j=1}^i \frac{Z_j}{n-j+1} + \frac{\{q(n-1)\}Z_{i+1}}{n-i}
\end{equation}

Here, $i=\lfloor q(n-1)\rfloor+1$. Noting that the true value of the $q$-quartile of the exponential distribution with rate $1$ is $\phi_q = -\log(1-q)$, we can get the bias of the seventh strategy:

\begin{align}
b_q &= \Bbb E(\hat{Q}_7(q))-\phi_q \notag \\
&= \left(\frac{1}{n}+\frac{1}{n-1}+\frac{1}{n-2}+\dots +\frac{1}{n-\lfloor q(n-1)\rfloor}\right)+\frac{\{q(n-1)\}}{n-\lfloor q(n-1)\rfloor-1}+\log(1-q)
\label{expon_bias}
\end{align}

This idea is implemented in the following Python routine:

\begin{lstlisting}[language=Python]
import numpy as np

def prcntile_bias_exponential(q, n):
	lt = int(np.floor(q*(n-1)))
	summ = 0
	for ix in range(lt+1):
		summ += 1/(n-ix)
	# modf is the fractional part
	summ += np.modf(q*(n-1))[0]/(n-lt-1)
	return -np.log(1-q)-summ
\end{lstlisting}

In figure \ref{fig:first}, we plot this bias with the quantile, $q$ for three estimators. The first one is $X_{(i)}$, the second is the linear interpolation estimator given by equation ~\ref{expon_bias} and the third is $X_{(i+1)}$ (with $i$ given by equation ~\ref{eqn:ord_i}). First, we see that the linear interpolation goes a long way in reducing the bias. 
Note that its bias starts off negative for low quantiles and becomes very high for high quantiles (we know it explodes to $\infty$ for the maximum, $X_{(n)}$ for distributions supported on $[a,\infty)$ for any $a$). Since it increases monotonically, starts below $0$ then goes to $\infty$ at $q=1$, there is one quantile for all exponential distributions which becomes unbiased. This happens to be extremely close to the $.6667$-th quantile (or the point that splits the distribution in a $1:2$ ratio).

\begin{figure}
\begin{center}
\includegraphics[height=3in, width=3in]{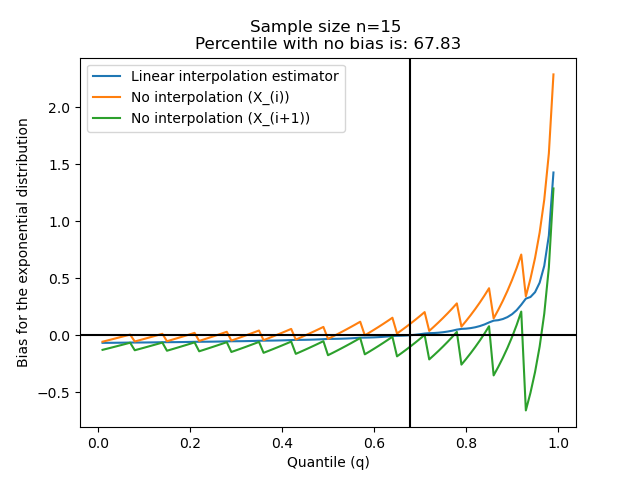}
\end{center}
\caption{The bias of the exponential distribution for linear interpolation and no interpolation estimators. \label{fig:first}}
\end{figure}

\section{Removing bias for the exponential}
\label{sec:verify}
We can use equation ~\ref{expon_bias} not just to measure the bias of strategy-7 on the exponential distribution, but also devise a new strategy that completely eliminates it. We can do this by replacing the linear interpolation strategy with another kind of interpolation specifically designed to eliminate the bias. Let's say the interpolation factor is $f \in \Bbb R$. We get: 

\begin{align}
b_q &= -\log(1-q)-\left(\frac{1}{n}+\frac{1}{n-1}+\frac{1}{n-2}+\dots +\frac{1}{n-\lfloor q(n-1)\rfloor}\right)-\frac{f}{n-\lfloor q(n-1)\rfloor-1} \notag\\
&= 0 \notag\\
=> f &= (n- \lfloor q(n-1)\rfloor-1)\left(-\log(1-q) - \left(\frac{1}{n-1}+\frac{1}{n-2}+\dots +\frac{1}{n-\lfloor q(n-1)\rfloor} \right)\right)
\end{align}

This idea is implemented in the following Python routine:

\begin{lstlisting}[language=Python]
import numpy as np

def expon_frac(q, n):
	lt = int(np.floor(q*(n-1)))
	summ = 0
	for ix in range(lt+1):
		summ += 1/(n-ix)
	return (-np.log(1-q)-summ)*(n-lt-1)
\end{lstlisting}

This new estimator based on removing the bias of the exponential distribution will be referred to as the `no exponential bias estimator' and is given by:

\begin{equation}
\hat{Q}_{10}(q) = f X_{(i)}+(1-f)X_{(i+1)}\label{eqn_q_10_estimator}
\end{equation}

where the $i$ is per equation ~\ref{eqn:ord_i}.

\section{Comparative performance on other distributions}
We've removed the bias for the exponential distribution. But we're more interested in heavy tailed distributions supported on $(0, \infty)$ with decreasing hazard rates. Two such distributions are the LogNormal distribution, which has a decreasing hazard rate (and heavier tailed than the exponential) and the Lomax distribution which is so heavy tailed, even it's mean blows up. In figures \ref{fig:biases} and \ref{fig:mses}, we compare estimators 4 and 7 with the `no exponential bias' estimator on the criterion of bias and MSE respectively for these two and other distributions (when the sample size is $n=15$).

The observations are as follows:

\begin{itemize}
\item{Estimator 4 has a tendency to have large negative biases for low quantiles. The other two estimators don't suffer from this problem}
\item{The `no exponential bias' estimator shows low bias on most distributions across quantiles less than about $0.6$. It does particularly well on low quantiles, always beating the other two. But even around the median, its performance is the most consistent. It performs the best on the LogNormal and comes in a close second on all others.}
\item{Apart from low quantiles, where the performance of estimator-4 really suffers, there isn't a lot of difference between the three.}
\end{itemize}

\begin{figure}
\begin{center}
\includegraphics[height=5in, width=5in]{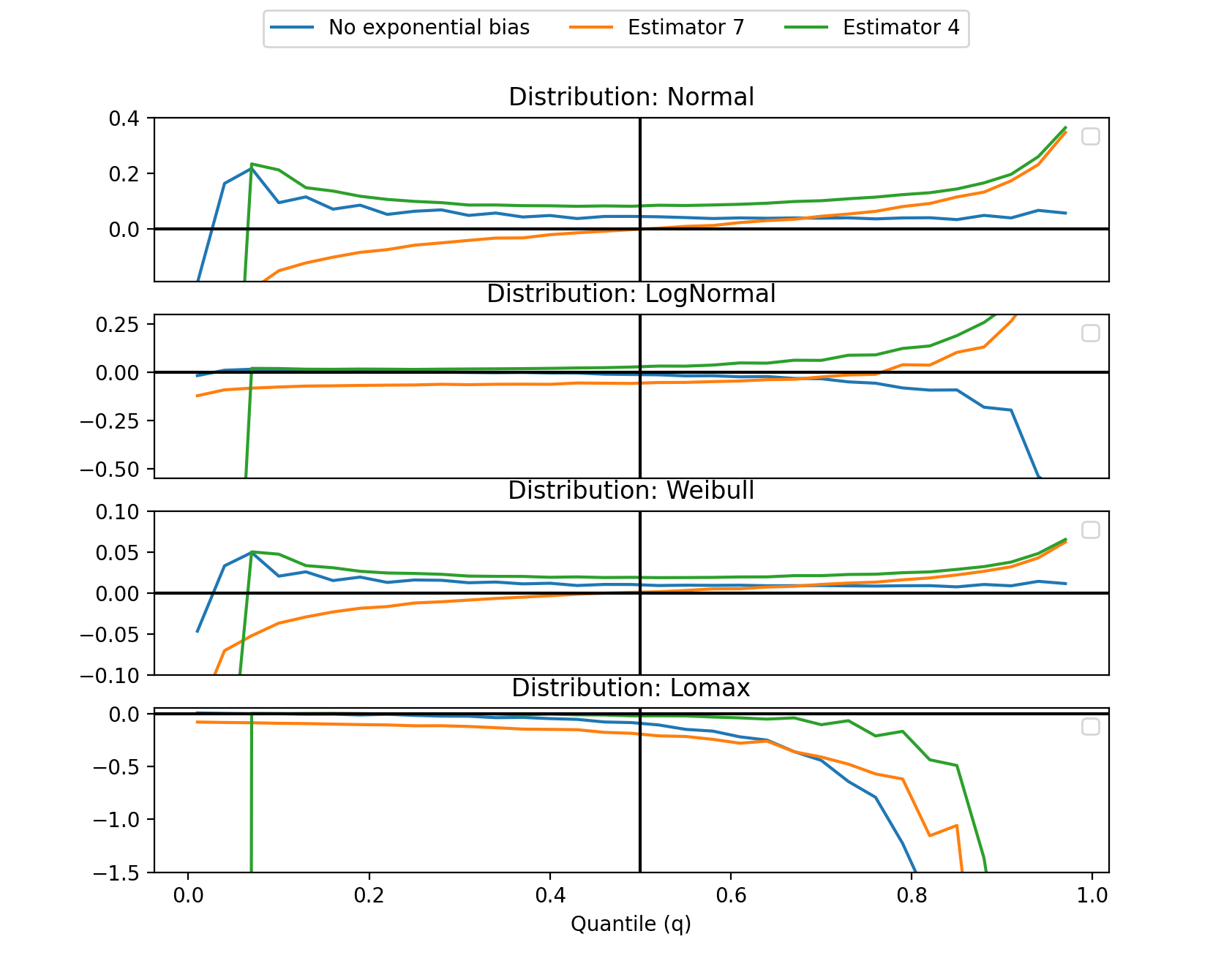}
\end{center}
\caption{The bias of the linear interpolation and no exponential bias estimators on various distributions for sample size $n=15$. \label{fig:biases}}
\end{figure}

\begin{figure}
\begin{center}
\includegraphics[height=5in, width=5in]{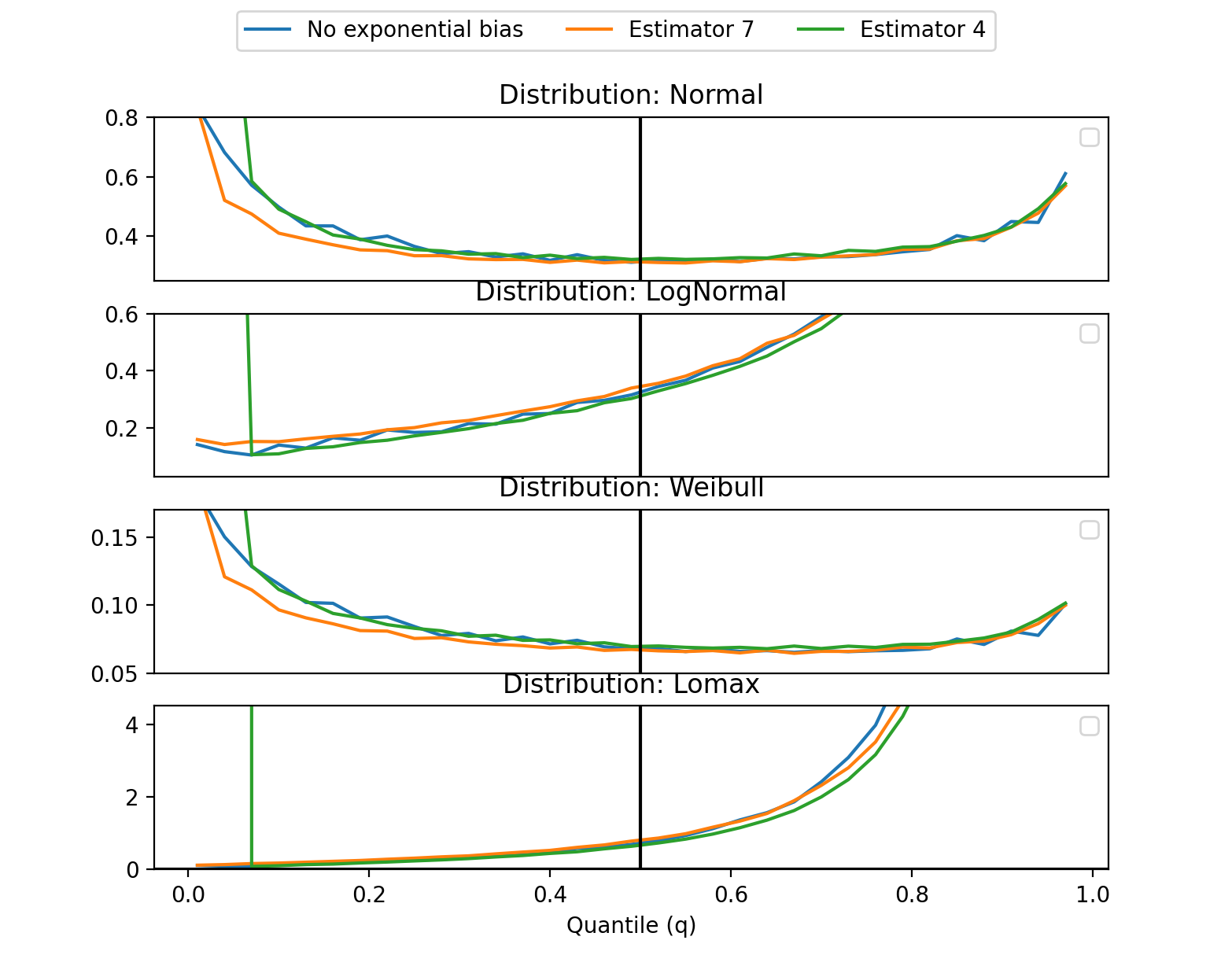}
\end{center}
\caption{The MSE of the linear interpolation and no exponential bias estimators on various distributions . \label{fig:mses}}
\end{figure}

\section{Generalizing the bias removal model}

The motivation of the previous section was to create an estimator that does well on the exponential distribution and then see how well it generalizes to heavy tailed distributions. In this section, we see if we can do even better on the exponential distribution by generalizing the estimator along two directions:

\begin{itemize}
\item{We took $i=\lfloor q(n-1)\rfloor+1$. Let's remove this requirement and just go with a general $i$. We can then ask what $i$ is the best.}
\item{We interpolated between the order statistics, $X_{(i)}$ and $X_{(i+1)}$. Why not include more of the order statistics beyond $X_{(i+1)}$ in our estimator?}
\end{itemize}

In practice, I've observed that going in this direction leads to an estimator that performs better on the exponential distribution, but there either isn't a significant improvement or a downright degradation in the performance on other distributions. 

This is how we define our estimator:

\begin{equation}
\hat{Q}_{11}(q) = f X_{(i)}+\sum\limits_{j=i+1}^m f_j X_{(j)}\label{general_estimator}
\end{equation}

where

$$f + \sum\limits_{j=i+1}^m f_j = 1$$

Note that this idea of including more quantiles is also explored in \cite{sheather}.

For the exponential distribution, we can plug in equation ~\ref{exp_ord} and get:

$$\hat{Q}_{11}(q) = \frac{1}{\lambda}\left(\sum\limits_{j=1}^i \frac{Z_j}{n-j+1} + \sum\limits_{k=1}^m \left(\sum\limits_{l=k}^m f_l\right) \frac{Z_{i+k}}{n-i-k+1} \right)$$

Taking the expectation this gives us:

$$\Bbb E(\hat{Q}_{11}(q)) = \frac{1}{\lambda}\left(\sum\limits_{j=1}^i \frac{1}{n-j+1} + \sum\limits_{k=1}^m  \frac{\sum\limits_{l=k}^m f_l}{n-i-k+1} \right)$$

And the variance of the estimator:

\begin{equation}
\Bbb V(\hat{Q}_{11}(q)) = \frac{1}{\lambda^2}\left(\sum\limits_{j=1}^i \frac{1}{(n-j+1)^2} + \sum\limits_{k=1}^m  \frac{\left(\sum\limits_{l=k}^m f_l\right)^2}{(n-i-k+1)^2} \right)
\label{eqn:variance}
\end{equation}

Now, we'd like to pick our $f$ and $f_j$'s in a way that we keep the estmator unbiased for the exponential distribution (constraint) while minimizing the variance of the estimator (objective function). Removing the terms in the variance that don't depend on the $f$ terms and constant mulipliers leads to the optimization problem:

\begin{equation*}
\begin{aligned}
& \underset{f, f_1, f_2, \dots f_m}{\text{min}}
& & \sum\limits_{k=1}^m \left(\sum\limits_{l=k}^m f_l\right)^2 \frac{1}{(n-i-k+1)^2}  \\
& \text{subject to}
& & \sum\limits_{j=1}^n \frac{1}{n-j+1} + \sum\limits_{k=1}^m \left(\sum\limits_{l=k}^m f_l\right) \frac{1}{n-i-k+1}  = -\log(1-q)
\end{aligned}
\end{equation*}
 
We can use the method of Lagrange multipliers to convert this optimization problem into a system of equations. Let's call $\beta$ the Lagrange multiplier for our single constraint (the zero bias condition). The Lagrangian becomes:

$$L(\beta, \vec{f}) = \sum\limits_{k=1}^m \frac{\left(\sum\limits_{l=k}^m f_l\right)^2 }{(n-i-k+1)^2} \\+ \beta \left(\sum\limits_{j=1}^i \frac{1}{n-j+1} + \sum\limits_{k=1}^m \frac{\sum\limits_{l=k}^m f_l}{n-i-k+1} +\log(1-q)\right)$$

The system of equations we need to solve are obtained by taking derivatives with respect to all the $f$'s and $\beta$ and setting them to zero. Taking the derivative with respect to $f_1$ first,

\begin{align}
\frac{2 \sum\limits_{j=1}^m f_j}{(n-i)^2} +\frac{\beta}{(n-i)} &=0\notag \\
=>\sum\limits_{j=1}^m f_j &= \frac{-\beta (n-i)}{2}\label{eqn:f1}
\end{align}

And now taking the derivative with respect to $f_2$,

\begin{align}
\frac{2 \sum\limits_{j=1}^m f_j}{(n-i)^2} + \frac{2 \sum\limits_{j=2}^m f_j}{(n-i-1)^2}+\frac{\beta}{(n-i)}+\frac{\beta}{(n-i-1)} &=0\notag \\
=>\sum\limits_{j=2}^m f_j &= \frac{-\beta (n-i-1)}{2}\label{eqn:f2}
\end{align}

Note that in equation ~\ref{eqn:f2}, we used ~\ref{eqn:f1} to cancel terms. Subtracting the two equations above we get $f_1 = -\frac{\beta}{2}$.

And taking derivatives similarly with respect to the other $f_l$ we'll get in general:

\begin{equation}
\sum\limits_{j=l}^m f_j = \frac{-\beta (n-i-(l-1))}{2}\label{eqn:sum_f}
\end{equation}

subtracting successive equations, all the $f_l$'s simply become $-\frac{\beta}{2}$ except the very last one:

\begin{align}
f_1, f_2, \dots f_{m-1} = -\frac{\beta}{2} \notag \\
f_m =  -\frac{\beta (n-i-(m-1))}{2} \label{the_fs}
\end{align}

And taking derivative with respect to $\beta$ and using equation ~\ref{the_fs},

\begin{align}
\sum\limits_{j=1}^n \frac{1}{n-j+1} -\frac{m \beta}{2} +\log(1-q) = 0 \notag \\
=> \beta = \left(\sum\limits_{j=1}^n \frac{1}{n-j+1}+ \log(1-q)\right)\frac{2}{m} = \frac{2b}{m}
\end{align}

Here, $b$ is the bias if we use just $X_{(i)}$ as the estimator. From equations ~\ref{eqn:sum_f} and ~\ref{eqn:variance} we get the variance of this estimator (the bias is $0$ by design):

\begin{align}
\Bbb V(\hat{Q}_{11}(q)) &= \sum\limits_{j=1}^i \frac{1}{(n-j+1)^2}+\frac{m \beta^2}{2} \notag\\
 &= \sum\limits_{j=1}^i \frac{1}{(n-j+1)^2}+\frac{b^2}{m} \label{eqn:var}
\end{align}

In other words, the higher we make $m$, the lower the variance becomes. This suggests that we should include all possible $m=n-i$ terms for the lowest possible variance (for the exponential distribution). If we do this, we get conveniently (from equation ~\ref{the_fs}):

$$f_1, f_2, \dots f_{m} = -\frac{\beta}{2}$$

The next logical question is, what $i$ one should pick? To minimize the variance given by equation ~\ref{eqn:var}, it turns out that the lower the $i$, the lower the variance. And the lowest possible value $i$ can take is $0$, with $X_{(0)}$ interpreted as the minimum possible value the exponential distribution can take ($0$). Since we intend to take all the $f_l$'s, we know that all of them are $-\frac{\beta}{2}$. Plugging $i=0$ and $m=n$ into equation ~\ref{general_estimator} we get:

\begin{align*}
\hat{Q}_{11}(q) &= f X_{(0)}+\sum\limits_{j=1}^n f_j X_{(j)}\\
&= -\frac{\beta}{2} \sum\limits_{j=1}^n X_{(j)}\\
&= -\left(\sum\limits_{j=1}^n \frac{1}{n-j+1} +\log(1-q)\right) \frac{\sum\limits_{j=1}^n X_{j}}{n}
\end{align*}

The ability of the estimator to remove the bias has been compromized since we took $X_{(0)}=0$. But we can just make the term in the equation above unbiased again and get:

$$\hat{Q}_{11}(q) = -\frac{\sum X_j}{n} \log(1-q)$$

This is equivalent to taking the MLE (maximum likelihood estimator) of the exponential distribution on the data and returning its relevant quantile. And while its true that this estimator does very well on the exponential distribution, with no bias and the smallest variance among all other estimators, it doesn't generalize well to other distributions. It is therefore a case of `overfitting' to the exponential distribution.

\section{Conclusion}
\label{sec:conc}
It seems like creating an estimator that works well for the exponential distribution generalizes to other distributions that might more realistically model our data as well. The `model penalty', which is what we pay for assuming the wrong distribution depends on the quantiles were estimating. If we're estimating low percentiles up until the 60th (hence encompassing the median), the penalty is low and our estimator generalizes well. If we're estimating the tail however the model penalty becomes high and it becomes counter-productive to use the more complicated `no exponential bias' estimator. 

Now, the best estimator for the exponential distribution involves taking the maximum likelihood estimate of the rate parameter (on the observed samples) and then calculating its quantile. While it is true that this estimator performs better on the exponential distribution (also zero bias but much lower standard deviation), it doesn't generalize well to heavier tailed distributions.

\bigskip
\begin{center}
{\large\bf SUPPLEMENTARY MATERIAL}
\end{center}

\begin{description}

\item[Python-package for  the new quantile estimation method:] Python-package \href{https://github.com/ryu577/statest}{statest} contains code to perform the diagnostic methods described in the article.

\end{description}

\end{document}